\documentclass[reprint, amsmath,amssymb,showpacs, longbibliography, aps, prb, twocolumn, superscriptaddress,nofootinbib]{revtex4-2}
\usepackage{CJKutf8}

\usepackage{amsthm}
\newtheorem{theorem}{Theorem}

\newtheorem{example}[theorem]{Example}

\usepackage{multirow}
\usepackage[utf8]{inputenc}
\usepackage{amsmath,amssymb,amsfonts,stmaryrd}
\usepackage{physics}
\usepackage{dsfont}
\usepackage{graphicx}
\usepackage{xcolor}
\usepackage{colortbl}
\usepackage{graphicx}
\usepackage{cancel}
\usepackage{natbib}
\usepackage{color}
\usepackage{tikz}
\usepackage{mathrsfs}
\usepackage{relsize}

\usepackage[all]{xy}
\usepackage{yhmath}
\usepackage{blkarray}
\usepackage{tabularx}
\usepackage{array}
\usepackage{makecell}
\usepackage[percent]{overpic}

\usepackage{diagbox}
\usepackage{algorithm}
\usepackage{algorithmic}
\usepackage{lipsum}
\newcommand*\colvec[3][]{
    \begin{pmatrix}\ifx\relax#1\relax\else#1\\\fi#2\\#3\end{pmatrix}
}

\usepackage{amsmath}
\usepackage{bm} 
\usepackage{subfigure} 
\usepackage{environ}

\NewEnviron{eqs}{%
\begin{equation}\begin{aligned}
    \BODY
\end{aligned}\end{equation}}

\NewEnviron{eqs*}{%
  \begin{equation*}
  \begin{aligned}
    \BODY
  \end{aligned}
  \end{equation*}
}

\usepackage{url}
\usepackage{scalerel}
\usepackage{stackengine,wasysym}

\usepackage{hyperref}
\hypersetup{colorlinks=true,citecolor=blue,linkcolor=blue, urlcolor=blue, breaklinks=true}
\hypersetup{linktocpage}

\allowdisplaybreaks

\usepackage{mathtools}

\renewcommand{\tilde}{\widetilde}

\renewcommand{\leq}{\leqslant}

\graphicspath{{./Figures/}}

\newcolumntype{C}{>{\centering\arraybackslash}p{0.72em}}

\newcommand{\prlsection}[1]{\vspace{6pt}\noindent\textbf{\textsc{#1.}}—}

\usepackage{amsmath,eqparbox}
\newcommand{\mcol}[2]{\eqmakebox[#1][c]{$#2$}} 

\begin{document}

\author{Zijian Liang}
\affiliation{International Center for Quantum Materials, School of Physics, Peking University, Beijing 100871, China}

\author{Yu-An Chen}
\email[E-mail: ]{yuanchen@pku.edu.cn}
\affiliation{International Center for Quantum Materials, School of Physics, Peking University, Beijing 100871, China}

\date{\today}
\title{Self-dual bivariate bicycle codes with transversal Clifford gates}

\begin{abstract}

Bivariate bicycle codes are promising candidates for high-threshold, low-overhead fault-tolerant quantum memories. Meanwhile, color codes are the most prominent self-dual CSS codes, supporting transversal Clifford gates that have been demonstrated experimentally. In this work, we combine these advantages and introduce a broad family of self-dual bivariate bicycle codes. These codes achieve higher encoding rates than surface and color codes while admitting transversal CNOT, Hadamard, and $S$ gates.
In particular, we enumerate weight-8 self-dual bivariate bicycle codes with up to $n \leq 200$ physical qubits, realized on twisted tori that enhance code distance and improve stabilizer locality. Representative examples include codes with parameters $[[n,k,d]]$: $[[16,4,4]]$, $[[40,6,6]]$, $[[56,6,8]]$, $[[64,8,8]]$, $[[120,8,12]]$, $[[152,6,16]]$, and $[[160,8,16]]$.

\end{abstract}

\maketitle



\prlsection{Introduction}
Fault-tolerant quantum computation relies on quantum error correction~\cite{Shor1995Scheme, Steane1996Error, Knill1997Theory, Calderbank1997Quantum, gottesman1997stabilizer, kitaev2003fault}. 
The Kitaev toric code stands out for its high threshold and rapid experimental progress~\cite{bravyi1998quantum, dennis2002topological, semeghini2021probing, Verresen2021PredictionTC, breuckmann2021quantum, bluvstein2022quantum, google2023suppressing, Google2023NonAbelian, Google2024surface, iqbal2023topological, iqbal2024NonAbelian}. 
Meanwhile, bivariate bicycle (BB) codes on small tori have exhibited markedly improved performance, often by nearly an order of magnitude over the toric code, while retaining low-density parity-check (LDPC) structure~\cite{Bravyi2024HighThreshold, wang2024coprime, Wang2024Bivariate, tiew2024low, wolanski2024ambiguity, gong2024toward, maan2024machine, cowtan2024ssip, shaw2024lowering, cross2024linear, voss2024multivariate, liang2024operator, berthusen2025toward, eberhardt2024logical, lin2025single, liang2025generalized, chen2025anyon, liang2025planar, leroux2025romanescocodesbiastailoredqldpc}. 
Because high-distance LDPC codes strongly suppress logical error rates below threshold, they can provide comparable protection using substantially fewer physical qubits~\cite{Bravyi2024HighThreshold}.

Color codes~\cite{bombin2006topological} in two dimensions provide a complementary approach: as self-dual CSS codes equivalent to two copies of the toric code via folding~\cite{kubica2015unfolding}, they inherit topological protection and, by self-duality, support transversal Clifford gates.
The color codes have been demonstrated experimentally, including implementations with open boundaries~\cite{Lacroix2025ColorCode} and tests on a small torus using the $[[18,4,4]]$ color code~\cite{wang2025demonstration}.

In this work, we introduce \emph{self-dual bivariate bicycle (BB) codes} that unify these approaches by combining high encoding rates, topological protection, and transversal Clifford gates. A systematic analysis and comprehensive search of self-dual BB codes is presented. It is shown that weight-6 self-dual BB codes reduce to multiple copies of the color code, which motivates a focus on weight-8 stabilizers. In the weight-8 family, the stabilizers are doubly even (i.e., have weight divisible by $4$), ensuring a transversal $S$ gate, self-duality guarantees a transversal Hadamard gate, and, as in any CSS code, CNOT is transversal between two code blocks~\cite{gottesman1997stabilizer}.


We enumerate weight-8 self-dual BB codes with parameters $[[n,k,d]]$, where $n$ is the number of physical qubits, $k$ the number of logical qubits, and $d$ the code distance. For each $n\le 200$, we report the code with the largest $kd^2/n$ (motivated by the Bravyi--Poulin--Terhal bound~\cite{bravyi2009no,Bravyi2010Tradeoffs}), as summarized in Tables~\ref{tab: n_k_d 1} and~\ref{tab: n_k_d 2}.\footnote{We report only cases with $k>4$; the $k=2$ and $k=4$ cases are already realized by the Kitaev toric code and the color code, respectively.}
Representative examples include the self-dual BB codes $[[16,4,4]]$, $[[40,6,6]]$, $[[56,6,8]]$, $[[64,8,8]]$, $[[120,8,12]]$, $[[152,6,16]]$, and $[[160,8,16]]$. To the best of our knowledge, these code constructions are new.

\begin{figure}[t]
\centering
\begin{overpic}[width=0.9 \linewidth]{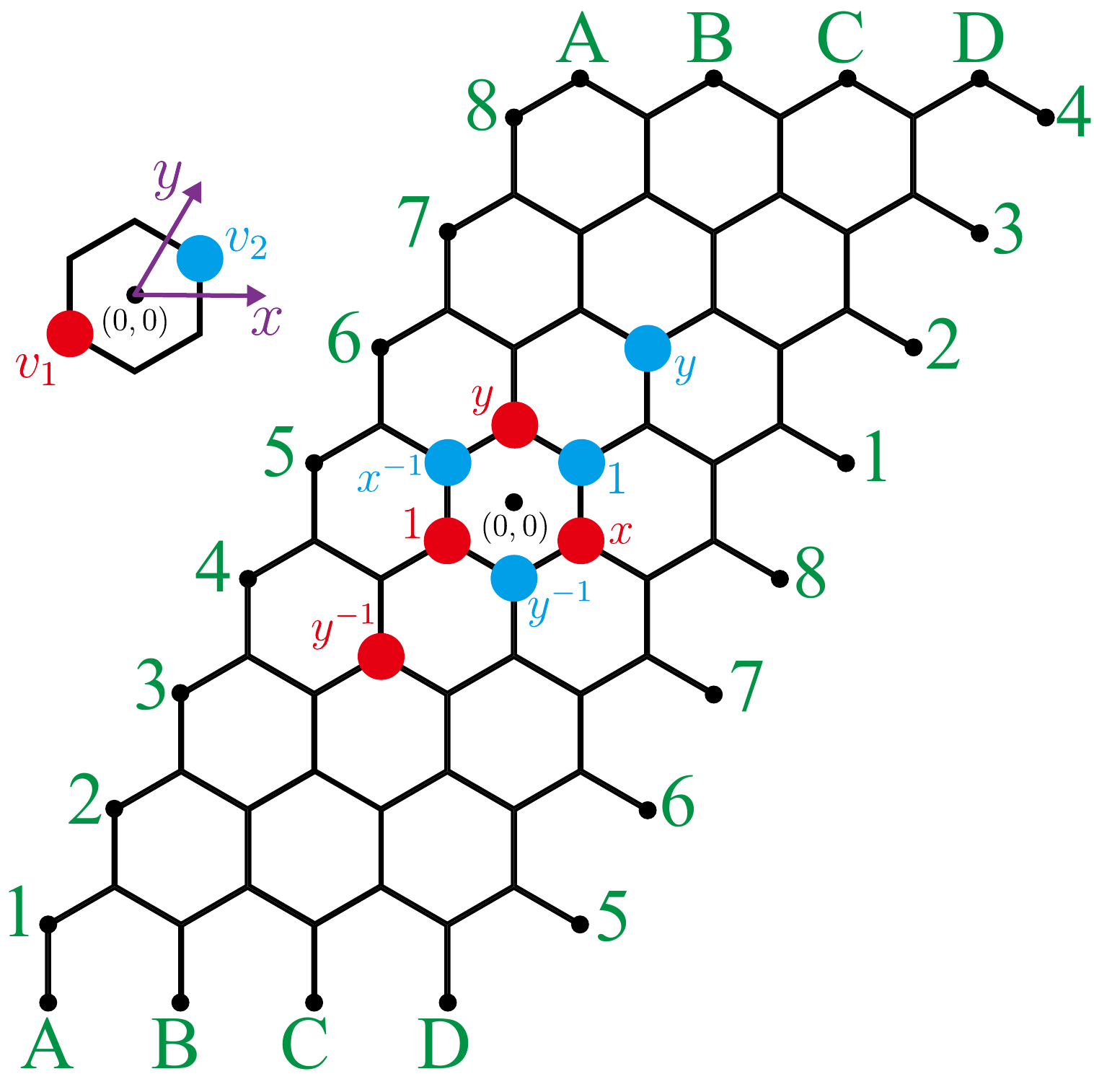}
     \put(65,5){\includegraphics[width=0.45\linewidth]{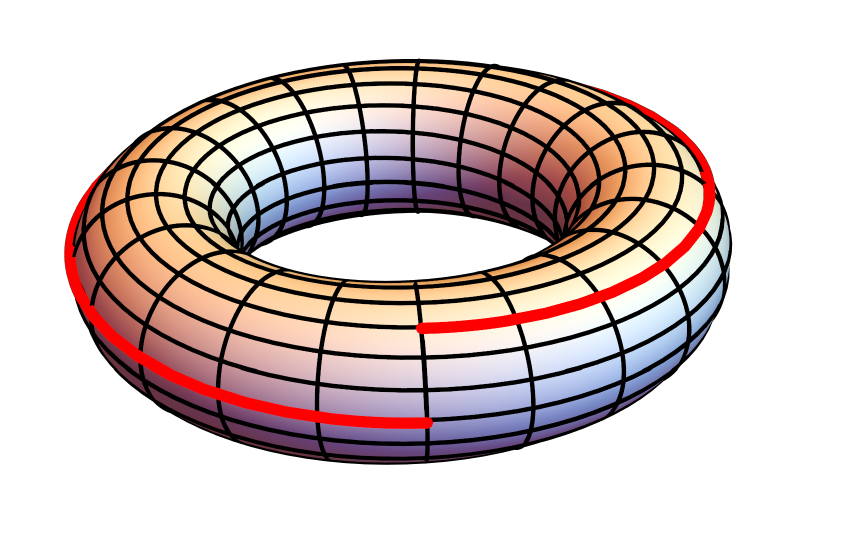}}
\end{overpic}
\caption{Stabilizer pattern of a weight-8 self-dual BB code on a twisted torus.     
    Red and blue dots indicate vertices defined by $f(x,y)=1+x+y+y^{-1}$ and $g(x,y)=\overline{f(x,y)} = 1+x^{-1}+y^{-1}+y$, respectively. 
    An $X$-stabilizer is the product of Pauli $X$ operators on all red and blue vertices (and similarly for $Z$). 
    The stabilizer possesses inversion symmetry about the plaquette center, ensuring commutation between $X$- and $Z$-stabilizers. 
    On a twisted torus with basis vectors $(0,8)$ and $(4,4)$, i.e., with boundary conditions $y^8=1$ and $x^4y^4=1$ (using the $x,y$ conventions of the top-left panel), this pattern realizes the $[[64,8,8]]$ code.
    The bottom-right panel shows the twisted torus with a longitudinal twist by a fraction of $2\pi$.
    The construction extends to arbitrary $f(x,y)$ and to twisted tori with basis vectors $\vec{a}_1$ and $\vec{a}_2$.}
\label{fig: self-dual BB code}
\end{figure}

\renewcommand{\arraystretch}{1.3}
\setlength{\tabcolsep}{0pt} 
\begin{table}[t]
\centering
\definecolor{mycolor1}{RGB}{255, 200, 100}  
\definecolor{mycolor2}{RGB}{200, 100, 200}
\definecolor{mycolor3}{RGB}{100, 180, 150}
\definecolor{mycolor4}{RGB}{100, 100, 150}
\definecolor{mycolor5}{RGB}{174, 217, 69}
\definecolor{mycolor6}{RGB}{250, 50, 200}
\definecolor{mycolor7}{RGB}{50, 250, 250}
\definecolor{mycolor8}{RGB}{250, 250, 50}
\begin{tabular}{|c|c|c|c|c|}
\hline
$[[n,k,d]]$ & $f(x,y)$ 
& $\vec{a}_1$   &$\vec{a}_2$   
&$\frac{kd^2}{n}$ 
\\ \hline

\rowcolor{mycolor8!60}$[[16 ,4 ,4 ]]$ & $1 + x + y + y^{-1}$&
$(0,4)$&$(2,2)$  
&4
\\ \hline
 
$[[24 ,8 ,4 ]]$ & $1 + x + x^{-1}y + xy$&
$(0,6)$&$(2,2)$   
&5.33
\\ \hline


\rowcolor{green!40}$[[30 ,6 ,5 ]]$ & $1 + x + x^2y + x^{-1}y$&
$(0,3)$&$(5,0)$   
&5
\\ \hline

\rowcolor{green!40}$[[32 ,12 ,4 ]]$ & $1 + x + x^2y + x^{-1}y$&
$(0,4)$&$(4,2)$   
&6
\\ \hline

\rowcolor{mycolor6!25} $[[36 ,10 ,4 ]]$ & $1 + x + x^{-1}y^{-1} + x^{-1}y $&
$(0,3)$&$(6,0)$   
&4.44
\\ \hline

$[[40 ,6 ,6 ]]$ & $1 + x + xy^{-1} + x^{-1}$&
$(0,4) $&$(5,1)$   
&\textbf{5.4}
\\ \hline

\rowcolor{green!40}$[[42 ,6 ,6 ]]$ & $1 + x + x^2y + x^{-1}y$&
$(0,3)$&$(7,0)$   
&5.14
\\ \hline

\rowcolor{mycolor7!70}$[[48 ,16 ,4 ]]$ & $1 + x + y^2 + x^{-2}$&
$(0,4)$&$(6,0)$   
&5.33
\\ \hline

\rowcolor{green!40}$[[50 ,10 ,5 ]]$ & $1 + x + x^2y + x^{-1}y$&
$(0,5)$&$(5,0)$   
&5
\\ \hline


\rowcolor{mycolor6!25} $[[54 ,10 ,6 ]]$& $1 + x + x^{-1}y^{-1} + x^{-1}y $&
$(0,9)$&$(3,3)$    
&6.67
\\ \hline

\rowcolor{green!40}$[[56 ,6 ,8 ]]$ & $1 + x + x^2y + x^{-1}y$&
$(0,7)$&$(4,3)$   
&\textbf{6.86}
\\ \hline

\rowcolor{green!40}$[[60 ,12 ,5 ]]$ & $1 + x + x^2y + x^{-1}y$&
$(0,6)$&$(5,0)$   
&5
\\ \hline
  
\rowcolor{mycolor8!60}$[[64 ,8 ,8 ]]$ & $1 + x + y + y^{-1}$&
$(0,8)$&$(4,4)$   
&\textbf{8}
\\ \hline

\rowcolor{mycolor7!70}$[[66 ,6 ,8 ]]$ & $1 + x + y^2 + x^{-2}$&
$(0,11)$&$(3,-3)$   
&5.82
\\ \hline


\rowcolor{green!40}$[[70 ,10 ,6 ]]$ & $1 + x + x^2y + x^{-1}y$&
$(0,5)$&$(7,0)$   
&5.14
\\ \hline

\rowcolor{green!40}$[[72 ,12 ,6 ]]$ & $1 + x + x^2y + x^{-1}y$&
$(0,6)$&$(6,3)$   
&6
\\ \hline


\rowcolor{mycolor7!70}$[[78 ,6 ,10 ]]$ & $1 + x + y^2 + x^{-2}$&
~$(0,13)$~ & ~$(3,-3)$~   
&~7.69~
\\ \hline

\rowcolor{green!40}$[[80 ,10 ,8 ]]$ & $1 + x + x^2y + x^{-1}y$&
$(0,8)$&$(5,4)$   
&8
\\ \hline

\rowcolor{mycolor6!25} $[[84 ,6 ,10 ]]$ & $1 + x + x^{-1}y^{-1} + x^{-1}y $&
$(0,21)$&$(2,5)$   
&7.14
\\ \hline

        
\rowcolor{red!35}$[[90 ,18 ,6 ]]$ & $1 + x + x^3y + x^{-1}y$&
$(0,15)$&$(3,6)$   
&7.2
\\ \hline


\rowcolor{mycolor6!25} $[[96 ,12 ,8 ]]$ & ~$1 + x + x^{-1}y^{-1} + x^{-1}y $~ &
$(0,12)$&$(4,4)$   
&8
\\ \hline

\rowcolor{green!40}$[[98 ,14 ,6 ]]$ & $1 + x + x^2y + x^{-1}y$&
$(0,7)$&$(7,0)$   
&5.14
\\ \hline

\rowcolor{green!40}~$[[100 ,12 ,8 ]]$~ & $1 + x + x^2y + x^{-1}y$&
$(0,5)$&$(10,0)$    
&7.68
\\ \hline

\rowcolor{mycolor7!70}$[[102 ,6 ,10 ]]$ & $1 + x + y^2 + x^{-2}$&
$(0,17)$&$(3,7)$     
&5.88
\\ \hline

\rowcolor{mycolor5!30}$[[104 ,6 ,12 ]]$ & $1 + x + y + x^{-1}y^{-1}$&
$(0,26)$&$(2,8)$  
&8.31
\\ \hline

$[[108 ,20 ,6 ]]$ & $1 + x + xy^2 + x^{-1}y^2$&
$(0,18) $&$(3,6)$   
&6.67
\\ \hline

\rowcolor{green!40}$[[110 ,10 ,8 ]]$ & $1 + x + x^2y + x^{-1}y$&
$(0,5)$&$(11,0)$   
&5.82
\\ \hline
\end{tabular}
\caption{Optimal weight-8 self-dual BB codes $[[n,k,d]]$ with $n \leq 110$. 
Each stabilizer code, defined by $f(x,y)=1+x+x^a y^b+x^c y^d$, is placed on a twisted torus specified by basis vectors $\vec{a}_1$ and $\vec{a}_2$, as illustrated in Fig.~\ref{fig: self-dual BB code}.
Rows with the same color correspond to identical stabilizers $f(x,y)$ realized on different tori. 
Code distances are computed exactly using the integer programming method of Refs.~\cite{landahl2011fault, Bravyi2024HighThreshold}. Boldface $kd^2/n$ exceeds all values at smaller $n$.
}
\label{tab: n_k_d 1}
\end{table}

\renewcommand{\arraystretch}{1.3}
\setlength{\tabcolsep}{0pt} 
\begin{table}[t]
\centering
\definecolor{mycolor1}{RGB}{255, 200, 100}  
\definecolor{mycolor2}{RGB}{200, 100, 200}
\definecolor{mycolor3}{RGB}{100, 180, 150}
\definecolor{mycolor4}{RGB}{100, 100, 150}
\definecolor{mycolor5}{RGB}{174, 217, 69}
\definecolor{mycolor6}{RGB}{250, 50, 200}
\definecolor{mycolor7}{RGB}{50, 250, 250}
\begin{tabular}{|c|c|c|c|c|}
\hline
$[[n,k,d]]$ & $f(x,y)$ 
& $\vec{a}_1$   &$\vec{a}_2$   
&$\frac{kd^2}{n}$ 
\\ \hline

\rowcolor{blue!30}$[[112 ,6 ,12 ]]$ & $1 + x + x^2y + x^{-1}y^2$&
$(0,7)$&$(8,0)$   
&7.71
\\ \hline

\rowcolor{mycolor7!70}$[[114 ,6 ,10 ]]$ & $1 + x + y^2 + x^{-2}$&
$(0,19)$&$(3,5)$    
&5.26
\\ \hline


$[[120 ,8 ,12 ]]$ & $1+x+y+x^{-2}y^{-2} $&
$(0,6)$&$(10,0)$    
&\textbf{9.6}
\\ \hline

     
$[[126 ,22 ,6 ]]$ & $1 + x + x^{-2}y + x^{-1}y^{-2}$&
$ (0,21) $&$ (3,6)$   
&6.29
\\ \hline

\rowcolor{green!40}$[[128 ,16 ,8 ]]$ & $1 + x + x^2y + x^{-1}y$&
$(0,8)$&$(8,4)$    
&8
\\ \hline

\rowcolor{red!35}$[[130 ,10 ,10 ]]$ & $1 + x + x^3y + x^{-1}y$&
$(0,5)$&$(13,1)$   
&7.69
\\ \hline

\rowcolor{brown!40}$[[132 ,8 ,12 ]]$ & $1 + x + y^2 + x^{-1}y^{-1}$&
$(0,33)$&$(2,11)$   
&8.73
\\ \hline

\rowcolor{blue!30}$[[136 ,6 ,14 ]]$ & $1 + x + x^2y + x^{-1}y^2$&
$(0,17)$&$(4,4)$  
&8.65
\\ \hline

\rowcolor{mycolor7!70}$[[138 ,6 ,12 ]]$ & $1 + x + y^2 + x^{-2}$&
$(0,23)$&$(3,5)$   
&6.26
\\ \hline

\rowcolor{green!40}$[[140 ,16 ,8 ]]$ & $1 + x + x^2y + x^{-1}y$&
$(0,7)$&$(10,0)$   
&7.31
\\ \hline

\rowcolor{mycolor4!30}$[[144 ,6 ,14 ]]$ & $1 + x + y + y^{-2}$&
$(0,24)$&$(3,11)$   
&8.17
\\ \hline


\rowcolor{mycolor6!25} $[[150 ,6 ,14 ]]$ & $1 + x + x^{-1}y^{-1} + x^{-1}y $&
$(0,15)$&$(5,7)$    
&7.84
\\ \hline

\rowcolor{blue!30}$[[152 ,6 ,16 ]]$ & $1 + x + x^2y + x^{-1}y^2$&
$(0,19)$&$(4,6)$  
&~\textbf{10.11}~
\\ \hline

\rowcolor{green!40}$[[154 ,14 ,8 ]]$ & $1 + x + x^2y + x^{-1}y$&
$(0,7)$&$(11,0)$  
&5.82
\\ \hline

\rowcolor{mycolor7!70}$[[156 ,12 ,10 ]]$ & $1 + x + y^2 + x^{-2}$&
$(0,26)$&$(3,10)$   
&7.69
\\ \hline
        
$[[160 ,8 ,16 ]]$ & $1+x+x^2y^2+x^{-1}y$&
$(0,10)$&$(8,0)$   
&\textbf{12.8}
\\ \hline

\rowcolor{mycolor4!30}$[[162 ,6 ,14 ]]$ & $1 + x + y + y^{-2}$&
$(0,27)$&$(3,12)$   
&7.26
\\ \hline

        
\rowcolor{mycolor5!30}$[[168 ,6 ,16 ]]$ & $1 + x + y + x^{-1}y^{-1}$&
$(0,42)$&$(2,10)$   
&9.14
\\ \hline

\rowcolor{red!35}$[[170 ,10 ,10 ]]$ & $1 + x + x^3y + x^{-1}y$&
$(0,5)$&$(17,0)$   
&5.88
\\ \hline


\rowcolor{mycolor7!70}$[[174 ,6 ,14 ]]$ & $1 + x + y^2 + x^{-2}$&
$(0,29)$&$(3,5)$    
&6.76
\\ \hline

\rowcolor{mycolor4!30}$[[176 ,8 ,16 ]]$ & $1 + x + y + y^{-2}$&
$(0,44)$&~$(2,30)$~   
&11.64
\\ \hline

\rowcolor{mycolor6!25} $[[180 ,10 ,12 ]]$ & $1 + x + x^{-1}y^{-1} + x^{-1}y $&
$(0,15)$&$(6,6)$    
&8
\\ \hline
        
$[[182 ,14 ,10 ]]$ & $1 + x + x^3y + x^{-2}y$&
$(0,7)$&$(13,0)$   
&7.69
\\ \hline

\rowcolor{blue!30}$[[184 ,6 ,16 ]]$ & $1 + x + x^2y + x^{-1}y^2$&
$(0,23)$&$(4,4)$  
&8.35
\\ \hline

\rowcolor{mycolor7!70}$[[186 ,6 ,14 ]]$ & $1 + x + y^2 + x^{-2}$&
$(0,31)$&$(3,5)$    
&6.32
\\ \hline


\rowcolor{red!35}$[[190 ,10 ,10 ]]$ & $1 + x + x^3y + x^{-1}y$&
$(0,5)$&$(19,0)$   
&5.26
\\ \hline
        
\rowcolor{brown!40}$[[192 ,12 ,12 ]]$ & $1 + x + y^2 + x^{-1}y^{-1}$&
~$(0,48)$~ &$(2,21)$   
&9
\\ \hline

\rowcolor{red!35}$[[196 ,14 ,10 ]]$ & $1 + x + x^3y + x^{-1}y$&
$(0,7)$&$(14,0)$   
&7.14
\\ \hline
       
\rowcolor{mycolor6!25} ~$[[198 ,10 ,12 ]]$~ & ~$1 + x + x^{-1}y^{-1} + x^{-1}y$~ &
~$(0,33)$~&$(3,9)$   
&7.27
\\ \hline

$[[200 ,12 ,12 ]]$ & $1 + x + x^{-1}y + xy^2$&
$(0,50)$&$(2,14)$   
&8.64
\\ \hline
\end{tabular}
\caption{Continuation of Table~\ref{tab: n_k_d 1} for $110 < n \leq 200$.}
\label{tab: n_k_d 2}
\end{table}

\prlsection{Algebraic–geometric methods}
We adopt an algebraic framework to analyze translation-invariant quantum codes on lattices~\cite{haah2016algebraic}. Ring-theoretic methods, in particular Gr\"obner basis techniques, simplify computations for CSS codes and provide a systematic classification of anyonic excitations and thereby determine the code properties~\cite{liang2025generalized, chen2025anyon}.

For concreteness, we focus on the honeycomb lattice, where one qubit is placed on each vertex. We briefly review the polynomial representation of Pauli operators~\cite{haah_module_13}. A unit cell contains two vertices, $v_1$ and $v_2$, whose Pauli operators are represented by four-dimensional vectors:
\renewcommand{\arraystretch}{1.3}
\begin{equation}
    \mathcal{X}_{v_1} =
    \begin{bmatrix}1 \\ 0 \\ \hline 0 \\ 0\end{bmatrix},~
    \mathcal{Z}_{v_1} =
    \begin{bmatrix}0 \\ 0 \\ \hline 1 \\ 0\end{bmatrix},~
    \mathcal{X}_{v_2} =
    \begin{bmatrix}0 \\ 1 \\ \hline 0 \\ 0\end{bmatrix},~
    \mathcal{Z}_{v_2} =
    \begin{bmatrix}0 \\ 0 \\ \hline 0 \\ 1\end{bmatrix}.
    \label{eq: X1 Z1 X2 Z2 definition}
\end{equation}
Translations by $(n,m)$ lattice units are implemented by multiplying the operator by $x^n y^m$ (with $n,m \in \mathbb{Z}$), as illustrated in the top-left panel of Fig.~\ref{fig: self-dual BB code}.
Products of Pauli operators correspond to sums of vectors, so the Pauli algebra forms a module over the Laurent polynomial ring
\begin{equation}
    R=\mathbb{Z}_2[x^{\pm1},y^{\pm1}].
\end{equation}
A translation-invariant CSS code is specified by two polynomials $f,g\in R$:
\renewcommand{\arraystretch}{1.35}
\begin{equation}
    S_X = 
    \begin{bmatrix} f(x,y) \\ g(x,y) \\ \hline 0 \\ 0 \end{bmatrix}, 
    \qquad
    S_Z = 
    \begin{bmatrix} 0 \\ 0 \\ \hline \overline{g(x,y)} \\ \overline{f(x,y)} \end{bmatrix},
    \label{eq: stabilizer}
\end{equation}
where $\overline{(\cdot)}$ denotes the antipode map $x^ny^m \mapsto x^{-n}y^{-m}$. All lattice translations of $S_X$ and $S_Z$ form the stabilizer group. A representative example of a self-dual BB code with $f(x,y)=1+x+y+y^{-1}$ and $g(x,y)=\overline{f(x,y)}$ is shown in Fig.~\ref{fig: self-dual BB code}.

The topological order condition is satisfied when $f(x,y)$ and $g(x,y)$ are coprime,
which ensures that any local operator commuting with all stabilizers is itself a stabilizer~\cite{eberhardt2024logical}.
Moreover, the maximal number of logical qubits on a torus equals the number of anyon generators in the underlying topological order~\cite{Witten1989Jones, Wen1995Topological, watanabe2023ground, liang2025generalized, zhang2025programmableanyon}. It can be computed as the codimension of the ideal generated by $f$ and $g$:
\begin{equation}
    k_\mathrm{max} = 2\, \dim \Bigg(
    \frac{\mathbb{Z}_2[x^{\pm1},y^{\pm1}]}{\langle f,\,g\rangle}
    \Bigg),
    \label{eq: maximal logical dimension}
\end{equation}
which is determined by the area of the corresponding Newton polytope~\cite{chen2025anyon}.
For a torus with twisted periodic boundary conditions specified by $\vec{a}_1 = (0, \alpha)$ and $\vec{a}_2 = (\beta,\gamma)$, the logical dimension is given by~\cite{liang2025generalized}
\begin{equation}
    k = 2 \dim \left(
    \frac{\mathbb{Z}_2[x^{\pm1},y^{\pm1}]}{\langle f(x,y),\,g(x,y),\,y^\alpha - 1,\,x^\beta y^\gamma - 1 \rangle}
    \right).
    \label{eq: k formula}
\end{equation}

\prlsection{Search for self-dual bivariate bicycle codes}
In what follows, we focus on self-dual bivariate bicycle codes satisfying $\overline{f}=g$, as illustrated in Fig.~\ref{fig: self-dual BB code}. On the lattice, this condition enforces inversion symmetry of each stabilizer about the center of the plaquette, ensuring that the $X$- and $Z$-stabilizers commute. For example, the standard color code on the honeycomb lattice corresponds to $f=1+x+y$, representing the product of Pauli $X$ or $Z$ operators on the six vertices of a hexagonal plaquette.

We first consider weight-6 stabilizers. The following theorem specifies their logical dimensions.
\begin{theorem}[Maximal logical dimension for self-dual BB codes]
\label{thm:bb-max-k}
Let $a,b,c,d \in \mathbb{Z}$, and define the polynomials
\begin{eqs}
    f(x,y) &= 1 + x^{a}y^{b} + x^{c}y^{d}, \\
    g(x,y) &= 1 + x^{-a}y^{-b} + x^{-c}y^{-d}~,
\end{eqs}
which generate the stabilizer~\eqref{eq: stabilizer} of a self-dual bivariate bicycle code. 
Setting $\Delta \coloneqq ad-bc$, the maximal number of logical qubits is
\begin{equation}
    k_{\max} \;=\; 4\,|\Delta| \, .
\end{equation}
Moreover, this maximum is realized on the twisted torus with periodicities specified by the basis vectors
\begin{equation}
    (3a,\,3b), \qquad (c+a,\,d+b),
    \label{eq: twisted TC a1 a2}
\end{equation}
i.e., the coordinate $(n_x,n_y)$ is identified with both $(n_x+3a,\,n_y+3b)$ and $(n_x+c+a,\,n_y+d+b)$.
\end{theorem}

\begin{example}
    As a concrete example, take the simplest choice $a=d=1$ and $b=c=0$, for which the stabilizers are $f=\overline{g}=1+x+y$, corresponding to the well-known color code~\cite{eberhardt2024pruning}. On the twisted torus with basis vectors $(3,0)$ and $(1,1)$, the construction yields the $[[6,4,2]]$ code with stabilizers
    \begin{equation}
    \begin{aligned}
    S_1 &= \mcol{c1}{X_1}\mcol{c2}{X_2}\mcol{c3}{X_3}\mcol{c4}{X_4}\mcol{c5}{X_5}\mcol{c6}{X_6},\\
    S_2 &= \mcol{c1}{Z_1}\mcol{c2}{Z_2}\mcol{c3}{Z_3}\mcol{c4}{Z_4}\mcol{c5}{Z_5}\mcol{c6}{Z_6}.
    \end{aligned}
    \end{equation}
    Alternatively, placing the same stabilizer polynomials on a $3\times 3$ torus with basis vectors $(3,0)$ and $(0,3)$ produces the $[[18,4,4]]$ code.
    More generally, for parameters $a,b,c,d$ with $|\Delta|>1$, the construction produces $|\Delta|$ independent copies of the color code, such as $[[12,8,2]]$ and $[[36,8,4]]$.
\label{example: color code}
\end{example}

\begin{figure*}[tbh]
    \centering
    \subfigure[~$\alpha$ logical pattern.]{\includegraphics[scale=0.25]{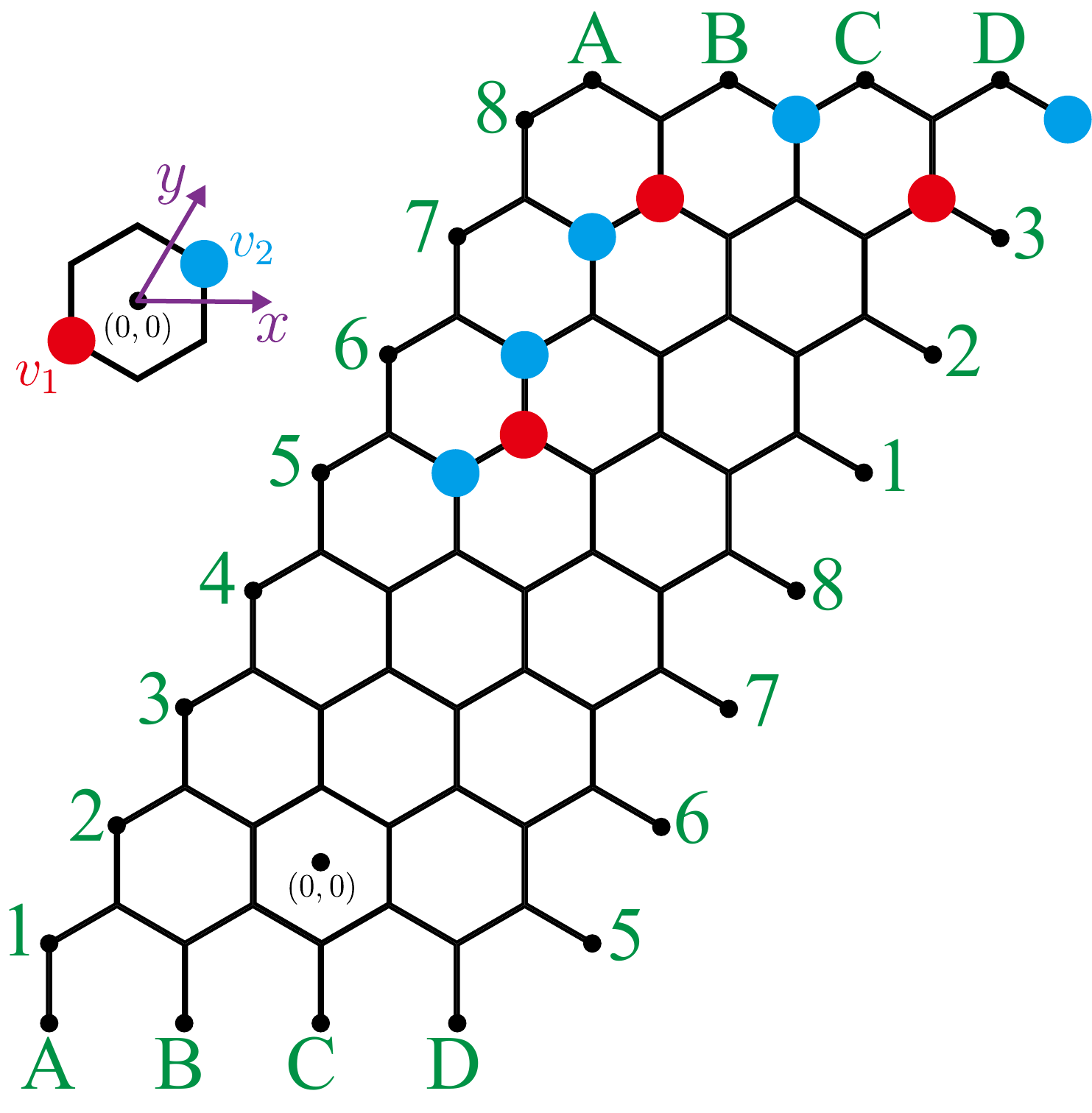}}
    \hspace{2em}
    \subfigure[~$\beta$ logical pattern.]{\includegraphics[scale=0.25]{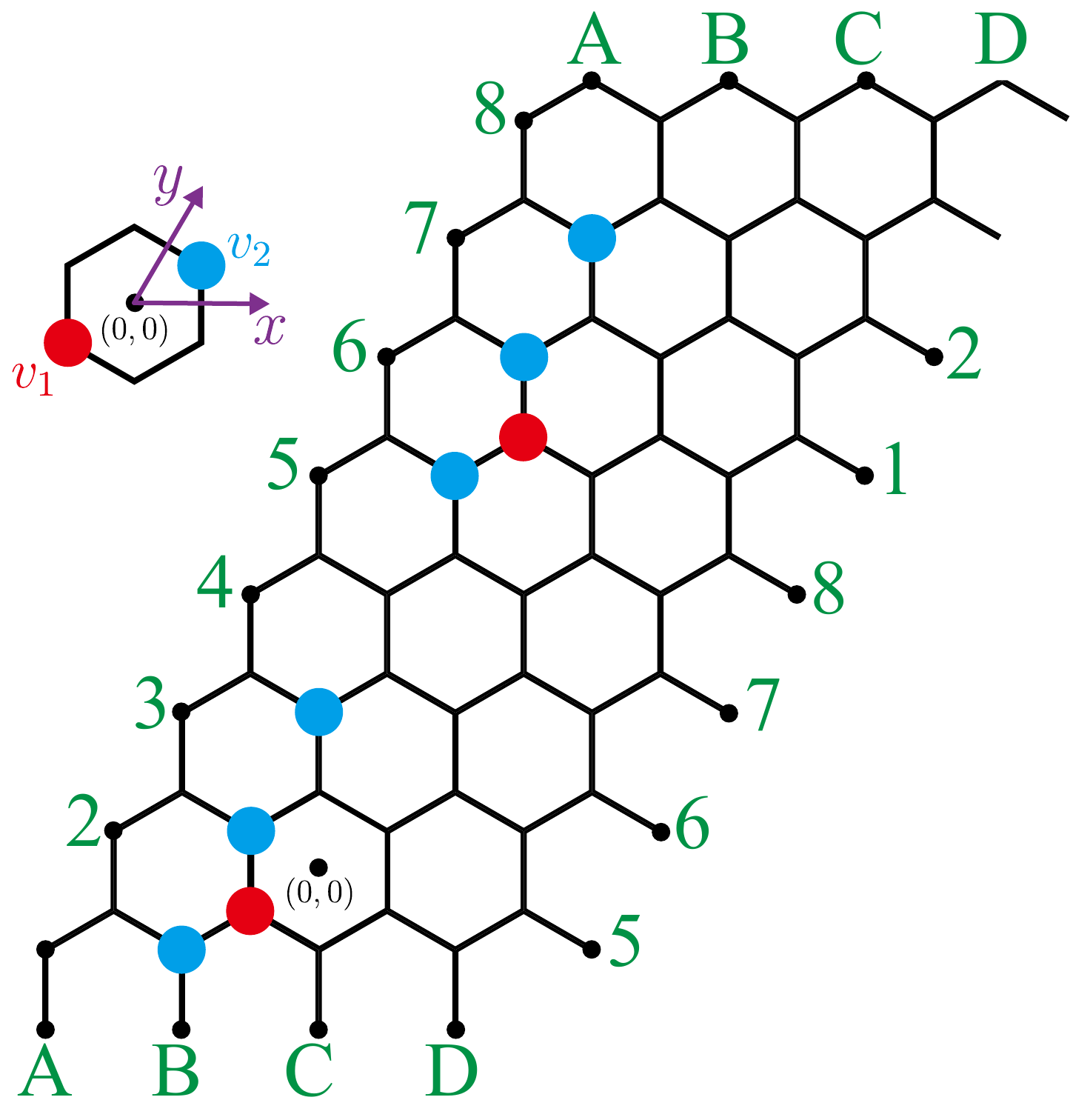}}
    \caption{For the $[[64,8,8]]$ code in Fig.~\ref{fig: self-dual BB code}, we show two representative support patterns for logical operators. All logical operators are generated by placing Pauli $X$ or $Z$ on the vertices of either pattern and their translations. One explicit choice of logical $\overline X_{1},\ldots,\overline X_{8}$ and $\overline Z_{1},\ldots,\overline Z_{8}$ is listed in Table~\ref{tab: logical pair}.}
    \label{fig: [[64, 8, 8]] code}
\end{figure*}

\begin{proof}[Proof of Theorem~\ref{thm:bb-max-k}]
Introduce $M=x^{a}y^{b}$ and $N=x^{c}y^{d}$, so that
\begin{equation}
f = 1+M+N~, \qquad g = 1+M^{-1}+N^{-1}~.
\end{equation}
We work over the Laurent polynomial ring $R'=\mathbb{Z}_2[M^{\pm1},N^{\pm1}]$ with lexicographic order $N \succ M$.

\emph{Buchberger algorithm~\cite{Buchberger1965algorithm}:}
Let $I:=\langle f,g\rangle$. Over $\mathbb{Z}_2$, addition and subtraction coincide. We first compute
\[
M\big((Ng)+f\big)=N+M^2\in I.
\]
Therefore, $I=\langle f,\,N+M^2\rangle$. Moreover,
\[
f+(N+M^2)=M^2+M+1,
\]
so we could replace $f$ by $M^2+M+1$ and obtain the Gr\"obner basis candidate
\begin{equation}
\mathcal{G}=\{\,N+M^2,\;\;M^2+M+1\,\}.
\end{equation}
The leading terms of $\mathcal{G}$ are
\begin{equation}
    \operatorname{LT}(N+M^2)=N,\text{ and } \operatorname{LT}(M^2+M+1)=M^2~.
\label{eq: leading terms}
\end{equation}
The corresponding $S$-polynomial is
\begin{eqs*}
&~~S(N+M^2,\,M^2+M+1) \\
&= M^2(N+M^2)+N(M^2+M+1) 
= NM+N+M^4~,
\end{eqs*}
which reduces to zero using $N+M^2$ and $M^2+M+1$. Hence, $\mathcal{G}$ is a Gr\"obner basis of $I$.

Since $M^2+M+1 \in I$, multiplying by $M + 1$ gives  
\[
(M+1)(M^2+M+1) = M^3 + 1 \in I~.
\]
Likewise, starting from $N + M^2 \in I$, multiplying by $M$ and adding $M^3 + 1$ yields  
\[
M(N + M^2) + (M^3 + 1) = NM + 1 \in I~.
\]
Therefore, in the quotient ring $R/I$, one obtains $M^3 = 1$ and $NM = 1$.

\emph{Dimension count:}
For a zero–dimensional Laurent ideal $I\subset\mathbb{Z}_2[x^{\pm1},y^{\pm1}]$, we can compute the quantity
\begin{equation}
    \dim\left(\frac{\mathbb{Z}_2[x^{\pm1},y^{\pm1}]}{I}\right)
\end{equation}
via the Bernstein–Khovanskii–Kushnirenko (BKK) theorem~\cite{Bernshtein1975, Sturmfels2002}.
BKK identifies this dimension with the mixed area of the Newton polytopes of the initial ideal $\operatorname{in}(I)$, generated by the leading terms of any Gr\"obner basis of $I$.
In our case, the leading terms~\eqref{eq: leading terms} correspond to the lattice vectors $(c,d)$ and $(2a,2b)$, so the mixed area of the two Newton segments is the area of the parallelogram they span, $2|\Delta|$. Hence, by Eq.~\eqref{eq: maximal logical dimension}, $k_{\max}=4|\Delta|$. 
Finally, the relations $M^3=1$ and $NM=1$ determine the twisted–torus periodicities in Eq.~\eqref{eq: twisted TC a1 a2}, completing the proof.
\end{proof}

As illustrated in Example~\ref{example: color code}, the weight-6 self-dual BB codes reduce either to the color code or to multiple decoupled copies of it. To obtain improved code parameters, it is therefore necessary to move beyond this case and consider higher-weight stabilizers, beginning with weight-8.

For computational efficiency, we restrict attention to polynomials of the form
\begin{equation}
    f(x,y) = 1 + x + x^a y^b + x^c y^d,
    \label{eq:weight-8 f(x,y)}
\end{equation}
with $g(x,y):=\overline{f(x,y)}$ its antipode.

We performed a systematic search for all weight-8 self-dual BB codes, subject to twisted periodic boundary conditions, for system sizes up to $n \leq 200$. For each even $n$, we first identified all decompositions $n = 2l \times m$ and defined the twisted torus by the basis vectors $\vec{a}_1 = (0,m)$ and $\vec{a}_2 = (l,q)$ with $0 \leq q < m$. We then enumerated all polynomials of the form~\eqref{eq:weight-8 f(x,y)} with exponent pairs $(a,b)$ and $(c,d)$ lying inside the parallelogram spanned by $\vec{a}_1$ and $\vec{a}_2$. For each candidate, we computed the corresponding $k$ by Eq.~\eqref{eq: k formula}.
The evaluation of $k$ in Eq.~\eqref{eq: maximal logical dimension} is independent of system size, since the boundary conditions $y^m-1$ and $x^l y^q-1$ reduce modulo $f(x,y)$ and $g(x,y)$, leaving only remainder terms relevant to the calculation. For cases with $k>0$, we compute the code distance using the integer programming method of Refs.~\cite{landahl2011fault, Bravyi2024HighThreshold}, thereby obtaining the full $[[n,k,d]]$ parameters.

The results are summarized in Tables~\ref{tab: n_k_d 1} and~\ref{tab: n_k_d 2}. 
For each $n$, we identify the codes that maximize the metric $kd^2/n$.
When multiple codes achieve the same optimal $kd^2/n$, we select the one with the most local stabilizers to present.


\prlsection{Transversal Clifford gates}
To achieve universal fault-tolerant computation, the first step is to realize logical gates of quantum error-correcting codes.  
For any weight-8 self-dual BB code, the transversal CNOT, Hadamard, and $S$ gates act as logical Clifford operations:
\begin{enumerate}
    \item \emph{CNOT (between blocks):}  
    For any CSS code, the transversal operation $\mathrm{CNOT}^{\otimes n}$ between two code blocks $A$ and $B$ preserves the stabilizer group and acts pairwise on logical qubits as
    \begin{eqs}
        \overline X_j^{(A)}&\mapsto\overline X_j^{(A)}\overline X_j^{(B)},&
        \overline Z_j^{(B)}&\mapsto\overline Z_j^{(A)}\overline Z_j^{(B)},\\
        \overline X_j^{(B)}&\mapsto\overline X_j^{(B)},&
        \overline Z_j^{(A)}&\mapsto\overline Z_j^{(A)}. \nonumber
    \end{eqs}

    \item \emph{Hadamard (single block):}  
    For self-dual CSS codes, transversal $H^{\otimes n}$ leaves the stabilizer invariant while exchanging $X$- and $Z$-type operators.  
    When each block encodes a single qubit, the logical operators can be chosen such that transversal Hadamard performs an encoded Hadamard rotation.  
    In our case, since self-dual BB codes encode multiple logical qubits, this transversal operation implements a multi-qubit Clifford transformation.  
    In an appropriate logical basis, it acts as
    \[
        \qquad \overline X_{2j-1}\leftrightarrow\overline Z_{2j},\qquad  
        \overline Z_{2j-1}\leftrightarrow\overline X_{2j}, \quad  
        1 \le j \le \frac{k}{2},
    \]
    exchanging $X$ and $Z$ operators pairwise. 

    \begin{table}[t]
        \centering
        \renewcommand{\arraystretch}{1.5}
        \setlength{\tabcolsep}{4.5pt}
        \begin{tabular}{|c|c|c|c|c|c|c|c|c|}
        \hline 
        Pair index & 1 & 2 & 3 & 4 & 5 & 6 & 7 & 8 \\
        \hline
        Logical $\overline{X}$ & $\alpha$ & $y^2\beta$ & $y\alpha$ & $y^3\beta$ & $y^2\alpha$ & $\beta$ &  $y^3\alpha$ & $y\beta$  \\
        \hline
        Logical $\overline{Z}$ & $y^2\beta$ & $\alpha$ & $y^3\beta$ & $y\alpha$ & $\beta$ & $y^2\alpha$ &  $y\beta$ & $y^3\alpha$  \\
        \hline
        \end{tabular}
        \caption{Eight pairs of logical operators $(\overline X_i,\overline Z_i)$ for the $[[64,8,8]]$ code. Here $\overline X_1$ is the product of Pauli $X$ operators supported on the pattern $\alpha$ in Fig.~\ref{fig: [[64, 8, 8]] code}, and $\overline Z_1$ is the product of Pauli $Z$ operators supported on the pattern $y^2\beta$ (i.e., $\beta$ shifted by two lattice units in the $y$ direction). The remaining logical operators are defined analogously. In this basis, the transversal Hadamard acts as $\overline X_{2j-1}\leftrightarrow \overline Z_{2j}$ and $\overline Z_{2j-1}\leftrightarrow \overline X_{2j}$, while the transversal phase gate $S^{\otimes n}$ acts as $\overline X_{2j-1}\mapsto \overline X_{2j-1}\overline Z_{2j}$ and $\overline X_{2j}\mapsto \overline X_{2j}\overline Z_{2j-1}$, with all $\overline Z_i$ fixed.}
        \label{tab: logical pair}
    \end{table}

    \item \emph{Phase $S$ (single block):}  
    The phase gate satisfies $SXS^\dagger = Y = iXZ$.  
    If every $X$-type stabilizer has weight divisible by $4$ (i.e., the code is doubly even), transversal $S^{\otimes n}$ preserves the stabilizer group.  
    In the same logical basis, its action on encoded operators is
    \begin{eqs}
        \overline X_{2j-1} &\mapsto \overline X_{2j-1}\,\overline Z_{2j}, &
        \overline Z_{2j-1} &\mapsto \overline Z_{2j-1}, \\
        \overline X_{2j}   &\mapsto \overline X_{2j}\,\overline Z_{2j-1}, &
        \overline Z_{2j}   &\mapsto \overline Z_{2j}.
    \end{eqs}
\end{enumerate}
These transversal operations generate multi-qubit logical Clifford transformations, providing a subset of fault-tolerant Clifford gates for weight-8 self-dual BB codes.

Explicit logical operators for the $[[64,8,8]]$ code are shown in Fig.~\ref{fig: [[64, 8, 8]] code} and Table~\ref{tab: logical pair}; their algebraic derivation is given in Appendix~\ref{app: Logical operators of [[64,8,8]]}. Physically, these operators can be viewed as string operators that transport anyons around the nontrivial cycles of the torus, analogous to the toric and color codes. In this model there are four independent $e$-$m$ pairs (Fig.~\ref{fig: m_a_e_c_braiding}), explaining why $k=8$.

\begin{figure}[t]
\centering
    \includegraphics[width= \linewidth]{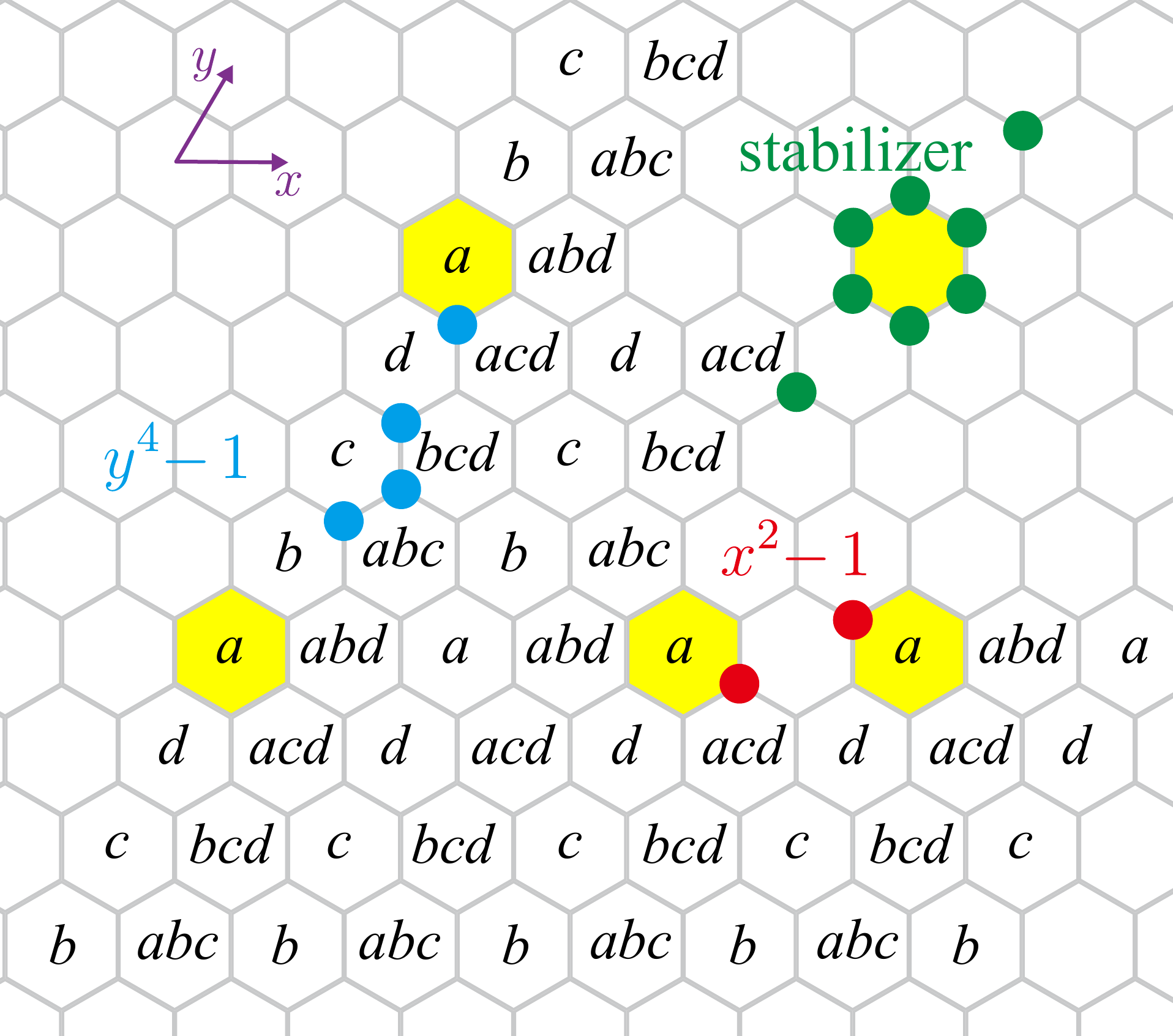}
    \caption{Mobility of anyons in the $[[64,8,8]]$ code. Yellow hexagons mark plaquettes supporting stabilizers (green dots). Violations of $Z$-type (respectively, $X$-type) stabilizers are identified as $e$ (respectively, $m$) anyons. We further refine these anyons into four flavors, labeled $a,b,c,d$; each hexagon is labeled by its flavor (or a fusion product). As in the color code, local string operators move an anyon only between hexagons of the same flavor. The red (blue) dots indicate the shortest string operators generating translations of anyons by two lattice units in $x$ (four units in $y$), corresponding to $x^2-1$ ($y^4-1$). Long strings wrapping the torus are obtained by concatenating these short strings and yield the logical-operator patterns $\alpha$ and $\beta$ in Fig.~\ref{fig: [[64, 8, 8]] code}.}
    \label{fig: m_a_e_c_braiding}
\end{figure}

\prlsection{Discussion}
A natural question is why we restrict the weight-8 search to the ansatz in Eq.~\eqref{eq:weight-8 f(x,y)} rather than the general
\begin{equation}
    f(x,y)=1+x^a y^b+x^c y^d+x^e y^f.
    \label{eq:general f(x,y)}
\end{equation}
There are two reasons. First, the general ansatz is substantially more expensive to scan: even with exponents bounded by $10$, varying the additional pair $(e,f)$ increases the search space by $\gtrsim 10^2$. Second, many instances of Eq.~\eqref{eq:general f(x,y)} are equivalent to Eq.~\eqref{eq:weight-8 f(x,y)} under an invertible monomial change of variables. For example, if $a$ and $b$ are coprime, then by B\'ezout's lemma there exist integers $a',b'$ with $ab'-ba'=\pm 1$. Then,
\begin{equation}
    x'=x^a y^b,\qquad y'=x^{a'} y^{b'},
\end{equation}
defines a unimodular change of basis and brings $f(x,y)$ into the form~\eqref{eq:weight-8 f(x,y)} in the variables $(x',y')$.
Thus Eq.~\eqref{eq:weight-8 f(x,y)} already covers a broad class of generic cases of Eq.~\eqref{eq:general f(x,y)}. Moreover, for $n=64$ we also tested the fully general ansatz~\eqref{eq:general f(x,y)} and observed no improvement in the metric $kd^2/n$.

We note that on the honeycomb lattice, any inversion-symmetric stabilizer pattern generates a self-dual BB code. As an example, the \emph{modified color code} introduced in Ref.~\cite{liang2023extracting} can be viewed as a dressed version of the standard hexagonal color code, obtained by adding inversion-symmetric vertices. The $[[64,8,8]]$ instance with $f = 1 + x + y + y^{-1}$ (referred to as the \emph{modified color code C} in Ref.~\cite{liang2023extracting}) naturally fits within this framework; however, its code parameters under periodic boundary conditions were not previously analyzed. The stabilizer defined by $f = 1 + x + y + y^{-1}$ represents the most local weight-8 construction, and given the rapid experimental progress in implementing color-code architectures, it is among the most promising candidates for near-term realization.

We also note that a code with parameters $[[64,8,8]]$ appears in Ref.~\cite{Ostrev2024IntersectingSubsets}, but that construction is not self-dual under a fixed qubit labeling (it becomes isomorphic only after column permutations). 
For comparison, a Reed-Muller CSS construction yields a self-dual $[[64,20,8]]$ code~\cite{PreskillNotesQEC}; however, its stabilizer generators have large weight (at least $16$) and the code is therefore not LDPC.

Self-dual BB codes combine high rate with symmetry and support fault-tolerant multi-qubit Clifford operations with low-weight checks and local geometry, making them a natural platform for near-term demonstrations of state preparation~\cite{iqbal2023topological,Google2024surface}, logical operations~\cite{Li2025TimeEfficient}, and repeated syndrome extraction. These attributes align with recent proposals for BB-style architectures and lattice surgery on superconducting hardware~\cite{yoder2025tour,yang2025planar}, suggesting a practical route to scalable fault-tolerant computation.

\section*{Acknowledgement}

We would like to thank Shin Ho Choe, Dongling Deng, Jens Niklas Eberhardt, Ying Li, Zhide Lu, Francisco Pereira, and Vincent Steffan for valuable discussions.

This work is supported by the National Natural Science Foundation of China (Grant No.~12474491), and the Fundamental Research Funds for the Central Universities, Peking University.

\bibliography{bib.bib}

\newpage

\appendix

\section{Logical operators of the $[[64,8,8]]$ code}\label{app: Logical operators of [[64,8,8]]}

In this appendix we derive explicit representatives of the logical operators for the $[[64,8,8]]$ code introduced in the main text. The code is defined on a twisted torus with boundary conditions
\[
x^4y^4=1,\qquad y^8=1,
\]
and stabilizer polynomials (written in a form convenient for Gr\"obner-basis computations)
\begin{eqs}
    f(x,y) &= (1+x+y+y^{-1})\,y,\\
    g(x,y) &= (1+x^{-1}+y^{-1}+y)\,x y.
\label{eq:3_-3_BB_f_g}
\end{eqs}
Let $R=\mathbb{Z}_2[x^{\pm1},y^{\pm1}]$ and $I:=\langle f,g\rangle\subset R$.

\paragraph{Anyon equation.}
Logical string operators correspond to solutions of the (generalized) anyon equation~\cite{liang2023extracting}
\begin{eqs}
  &(x^4 y^4 +1)\,b_x(x,y)\;+\;(y^{8}+1)\,b_y(x,y)\\
  &=p_{1}(x,y)\,g(x,y)\;+\;p_{2}(x,y)\,f(x,y),
\label{eq: (3, -3)-BB frustrated anyon equation}
\end{eqs}
for some $b_x,b_y,p_1,p_2\in R$. Equivalently, we seek $b_x,b_y\in R$ such that
\begin{equation}
    (x^4 y^4 +1)\,b_x(x,y)+(y^{8}+1)\,b_y(x,y)=0\;\in R/I,
\label{eq: (3, -3)-BB anyon equation mod I}
\end{equation}
with $I=\langle f,g\rangle$.

\paragraph{Gr\"obner basis and logical dimension.}
With monomial order $x>y$, a Gr\"obner basis of $I$ is
\begin{eqs}
    h(x,y) &= 1 + y^4,\\
    i(x,y) &= 1 + x + y + y^3,
\label{eq: (3, -3)-BB code Grobner basis 1}
\end{eqs}
where
\begin{eqs}
    h &= GB_{1,1}\,f + GB_{1,2}\,g, \\
    i &= GB_{2,1}\,f + GB_{2,2}\,g,
\end{eqs}
with coefficients
\begin{eqs}
    GB_{1,1} &= 1 + y + y^2, &\quad GB_{1,2} = y, \\
    GB_{2,1} &= 1 + y, &\quad GB_{2,2} = 1.
\end{eqs}
The standard monomials of $\langle h,i\rangle$ are $1,y,y^2,y^3$, so $\dim_{\mathbb{Z}_2}(R/I)=4$ and hence $k_{\max}=8$.  

Moreover,
\[
y^8+1=(1+y^4)\,h\in I,
\]
and $x^4y^4+1$ also reduces to zero modulo $\langle h,i\rangle$ (see Eq.~\eqref{eq: x4y4+1 into h and i} below). Therefore the boundary conditions do not further reduce the logical dimension, and the code achieves $k=8$.

\paragraph{Two independent cycle relations.}
Since both $x^4y^4+1$ and $y^8+1$ vanish in $R/I$, Eq.~\eqref{eq: (3, -3)-BB anyon equation mod I} is satisfied for any $(b_x,b_y)$. To obtain explicit logical representatives, it is convenient to extract two independent relations, corresponding to $(b_x,b_y)=(1,0)$ and $(0,1)$.

We first consider $(b_x,b_y)=(1,0)$, the string operator along the $x$-cycle.
Reducing $x^4y^4+1$ by the Gr\"obner basis~\eqref{eq: (3, -3)-BB code Grobner basis 1} gives
\begin{equation}
    x^4 y^4 +1 = q_h(x,y)\,h(x,y) + q_i(x,y)\,i(x,y)\pmod{2},
    \label{eq: x4y4+1 into h and i}
\end{equation}
with
\begin{eqs}
    q_h(x,y) &= x^4 + y^2 + x^2 y^2,\\
    q_i(x,y) &= 1 + x + x^2 + x^3 + y + x^2 y + y^3 + x^2 y^3.
    \nonumber
\end{eqs}
Combining this with the expressions of $h$ and $i$ in terms of $f$ and $g$ yields
\begin{equation}
    x^4 y^4 + 1 = p_1(x,y)\,g(x,y) + p_2(x,y)\,f(x,y)\pmod{2},
\end{equation}
where
\begin{eqs}
    p_1(x,y) &= q_h\,GB_{1,2} + q_i\,GB_{2,2} \\
             &= 1 + x + x^2 + x^3 + y + x^2 y + x^4 y,\\
    p_2(x,y) &= q_h\,GB_{1,1} + q_i\,GB_{2,1} \\
             &= 1 + x + x^2 + x^3 + x^4 + x y + x^3 y + x^4 y + x^4 y^2.
\label{eq: bx1_coeffs}
\end{eqs}

Next, we consider $(b_x,b_y)=(0,1)$, the string operator along the $y$-cycle.
Since $y^8+1=(1+y^4)\,h$, we immediately obtain
\begin{equation}
    1 + y^8 = p_1(x,y)\,g(x,y) + p_2(x,y)\,f(x,y)\pmod{2},
\end{equation}
with
\begin{eqs}
    p_1(x,y) &= (1 + y^4)\,GB_{1,2} = y + y^5,\\
    p_2(x,y) &= (1 + y^4)\,GB_{1,1} = 1 + y + y^2 + y^4 + y^5 + y^6.
\label{eq: by1_coeffs}
\end{eqs}

\paragraph{Return to the main-text conventions.}
The main text uses
\begin{eqs}
    f_0(x,y) &= 1 + x + y + y^{-1},\\
    g_0(x,y) &= 1 + x^{-1} + y^{-1} + y.
    \label{eq: original S}
\end{eqs}
Since Eq.~\eqref{eq:3_-3_BB_f_g} differs from $(f_0,g_0)$ by monomial prefactors (translations), the representatives in
Eqs.~\eqref{eq: bx1_coeffs} and~\eqref{eq: by1_coeffs} are converted to 
\begin{eqs}
    p_1'(x,y) := \frac{p_1(x,y)}{y}, \quad
    p_2'(x,y) := \frac{p_2(x,y)}{x y}.
    \label{eq: final logical}
\end{eqs}
Substituting the solution in Eq.~\eqref{eq: bx1_coeffs} into Eq.~\eqref{eq: final logical}, we obtain
\begin{eqs}
    p_1'(x,y) &= y^{-1} + x y^{-1} + x^2 y^{-1} + x^3 y^{-1} + 1 + x^2 + x^4,\\
    p_2'(x,y) &= x^{-1} y^{-1} + y^{-1} + x y^{-1} + x^2 y^{-1} + x^3 y^{-1} + 1 \nonumber\\
             &\quad + x^2 + x^3 + x^3 y~.
\end{eqs}
Up to multiplication by stabilizers, we may further simplify $(p_1',p_2')$ by adding a common polynomial multiple
of $(f_0,g_0)$:
\begin{equation*}
    (p_1',p_2') \sim \big(p_1' + s\,f_0,\; p_2' + s\,g_0\big)~, \quad \forall ~s(x,y)\in R~.
\end{equation*}
Choosing $s(x,y)=y^{-1}+x^2 y^{-1}$ yields an equivalent representative with lower weight,
\begin{eqs}
    \tilde{p}_1(x,y) :=& ~p_1'(x,y) + (y^{-1}+x^2 y^{-1})\, f_0(x,y) \\
    =& ~y^{-2} + x^2 y^{-2} + x^4~,\\
    \tilde{p}_2(x,y) :=& ~p_2'(x,y) + (y^{-1}+x^2 y^{-1})\, g_0(x,y) \\
    =& ~x^{3} y^{-1} + x^{3} + x^{3} y + y^{-2} + x^2 y^{-2}~.
\label{eq: bx1_coeffs 1}
\end{eqs}
This corresponds to the pattern $\alpha$ shown in Fig.~\ref{fig: [[64, 8, 8]] code}.
For example, according to Table~\ref{tab: logical pair}, the logical operator $\overline{X}_1$ is given by the product of $X$ operators supported on the pattern $\alpha$:
\begin{equation}
    \overline{X}_1 = \begin{bmatrix}\tilde{p}_1(x,y) \\ \tilde{p}_2(x,y) \\ \hline 0 \\ 0\end{bmatrix}.
\end{equation}

Similarly, substituting the solution in Eq.~\eqref{eq: by1_coeffs} into Eq.~\eqref{eq: final logical}, we find
\begin{eqs}
    p_1'(x,y) &=  1 + y^4,\\
    p_2'(x,y) &= x^{-1}y^{-1} + x^{-1} + x^{-1}y 
    \\
    & \quad +x^{-1}y^3 + x^{-1}y^4 + x^{-1}y^5~.
\label{eq: by1_coeffs 1}
\end{eqs}
This corresponds to the logical pattern $\beta$ shown in Fig.~\ref{fig: [[64, 8, 8]] code}.
For example, according to Table~\ref{tab: logical pair}, the logical operator $\overline{Z}_5$ is given by the product of $Z$ operators supported on the pattern $\beta$:
\begin{equation}
    \overline{Z}_5 = \begin{bmatrix}0 \\ 0 \\ \hline p_1'(x,y) \\ p_2'(x,y) \end{bmatrix}.
\end{equation}

Finally, we verify the commutation relations between these logical operators. A direct computation gives
\begin{eqs}
    &\overline{\tilde{p}_1(x,y)} ~p_1'(x,y) + \overline{\tilde{p}_2(x,y)} ~p_2'(x,y) \\
    &= y^2 + y^6 + x^2 (y^2 + y^6) \\
    &\quad+ x ~(y + y^2 + y^3 + y^5 + y^6 + y^7) \\
    &\quad+ x^3 (y + y^2 + y^3 + y^5 + y^6 + y^7)~,
\end{eqs}
where all coefficients are reduced modulo $2$, and we work modulo the ideal generated by the boundary conditions
$y^8-1$ and $x^4y^4-1$ (equivalently, we impose $y^8=1$ and $x^4=y^4$).
The first two terms indicate that an $X$ operator supported on $\alpha$ anticommutes with a $Z$ operator supported on $y^n\beta$ only for $n=2,6$. Therefore, the logical operators listed in Table~\ref{tab: logical pair} satisfy
\begin{equation}
    \overline{X}_i \overline{Z}_j = (-1)^{\delta_{i,j}} \overline{Z}_j \overline{X}_i~,
\end{equation}
as required.






\end{document}